%% file: threshold.tex
\documentclass{acm_proc_article-sp}

\usepackage{IEEEtrantools}
\usepackage{amsmath}
\usepackage{amssymb}
\usepackage{amsfonts}
\usepackage{array}
\usepackage{cite}

\newcommand{\bcomment}[1]{}
\newcommand{\TODO}[1]{}
\newcommand{\N}{\mathbb{N}}
\newcommand{\bfp}{\mathbf{p}}
\newcommand{\bfx}{\mathbf{x}}

\newcommand{\one}{\mathbf{1}}
\newcommand{\E}{\mathbb{E}}
\newcommand{\leqa}[1]{\stackrel{(#1)}{\leq}}
\newcommand{\eqa}[1]{\stackrel{(#1)}{=}}
\newcommand{\propa}[1]{\stackrel{(#1)}{\propto}}

\newcommand{\ER}{Erd\H{o}s-R\'{e}nyi }

\newtheorem{theorem}{Theorem}
\newtheorem{lemma}{Lemma}

\newtheorem{definition}{Definition}

\newtheorem{remark}{Remark}

\begin{document}

\title{Improved Achievability and Converse Bounds for \ER Graph Matching}

\numberofauthors{2} 
\author{
\alignauthor
Daniel Cullina\\
       \affaddr{University of Illinois at Urbana-Champaign}\\
       \affaddr{1308 W Main St}\\
       \affaddr{Urbana, Illinois, 61801}\\
       \email{cullina@illinois.edu}
\alignauthor
Negar Kiyavash\\
       \affaddr{University of Illinois at Urbana-Champaign}\\
       \affaddr{1308 W Main St}\\
       \affaddr{Urbana, Illinois, 61801}\\
       \email{kiyavash@illinois.edu}
}

\maketitle

\begin{abstract}
We consider the problem of perfectly recovering the vertex correspondence between two correlated \ER (ER) graphs.
For a pair of correlated graphs on the same vertex set, the correspondence between the vertices can be obscured by randomly permuting the vertex labels of one of the graphs.
In some cases, the structural information in the graphs allow this correspondence to be recovered.
We investigate the information-theoretic threshold for exact recovery, i.e. the conditions under which the entire vertex correspondence can be correctly recovered given unbounded computational resources.

Pedarsani and Grossglauser provided an achievability result of this type.
Their result establishes the scaling dependence of the threshold on the number of vertices.
We improve on their achievability bound.
We also provide a converse bound, establishing conditions under which exact recovery is impossible.
Together, these establish the scaling dependence of the threshold on the level of correlation between the two graphs.
The converse and achievability bounds differ by a factor of two for sparse, significantly correlated graphs.
\end{abstract}

\section{Introduction}
\label{section:intro}
In this paper we consider the problem of graph deanonymization, or graph matching.
In this problem, there are two correlated graphs on the same vertex set.
Call these $G_a$ and $G_b$.
By correlated we mean that the presence or absence of a particular edge in $G_a$ provides some information about the presence of that edge in $G_b$.
The correspondence between the vertices can be obscured by randomly permuting the vertex labels of $G_a$.
Given the permuted version of $G_a$ and the unaltered $G_b$, what can be learned about the correspondence between their vertex sets?

We focus on one particular variant: the problem of perfectly recovering the vertex correspondence between two correlated \ER (ER) graphs.
In some cases, the structural information in the graphs allow this correspondence to be recovered.
We investigate the information-theoretic threshold for exact recovery, i.e. the conditions under which the entire vertex correspondence can be correctly recovered given unbounded computational resources.

This question was first addressed by Pedarsani and Grossglauser, who provided an information-theoretic achievability result \cite{pedarsani_privacy_2011}.
Their result establishes the scaling dependence of the threshold on the number of vertices: in order for exact deanonymization to be feasible, $n$-vertex graphs $G_a$ and $G_b$ must have average degree at least $\Omega(\log n)$.
Our main result is a new achievability bound that improves on the bound from \cite{pedarsani_privacy_2011}. 
We also provide a converse bound, establishing conditions under which exact recovery is impossible.
Together, these establish the scaling dependence of the threshold on the level of correlation between the two graphs.
Call $G_a$ and $G_b$ sparse if their average degree is sublinear.
Call them significantly correlated if their intersection is larger than the intersection of similar independent graphs by a factor that grows to infinity.
In this regime, our converse and achievability bounds differ by a constant factor of two.

One motivation for the study of graph deanonymization comes from networks associated with internet services.
As these services have become ubiquitous, an enormous amount of data about the users of these services has been generated and collected.
Much of this data is \emph{structural}.
It reflects interactions between multiple users: communications from one user to another, personal relationships, transactions, and many more examples.
Other forms of data associated with single users are still have informative network structure.
Example of this type of data include home towns, employers, educational institutions, hobbies and interests, and purchase history.

This data allows unprecedented opportunities for analysis, particularly when multiple data sources are combined.
However, there are complex and difficult trade-offs between facilitating analysis and preserving privacy.

There are at least two fundamental reasons to attempt to learn vertex correspondences between networks.
First, if multiple networks reflect an common underlying network of relationships, then one can obtain a better estimate of the underlying network by combining multiple sources to overcome the effects of noisy data and omissions.
Second, if the data associated with one network is sensitive and the data associated with another network allows for the identification of users, then learning the vertex correspondence grant access to sensitive user information.

Responsible privacy management by data collectors requires an understanding of when sensitive information can and cannot be recovered from data.
A large portion of the recent work on graph deanonymization has involved the evaluation of heuristic algorithms on datasets derived from real-world networks.
These lines of work play a crucial role in advancing our understanding of the privacy of current systems.
We attempt to complement these efforts by contributing to a foundational theory that will inform the design of future system.

In addition to the practical motivations, this is an interesting and fundamental problem in theory of random graphs.
Throughout, we will discuss the connections to other questions regarding random graphs.

The remainder of the paper is organized as follows.
In Section~\ref{section:related-work}, we discuss some other work on information-theoretic limits of deanonymization.
In Section~\ref{section:model}, we introduce our notation and formalize the estimation problem and our model of correlated graphs.
In Section~\ref{section:results}, we state our main results.
Section~\ref{section:achievability} contains the proof of our main achievability result and Section~\ref{section:converse} contain the proof of our converse.
In Section~\ref{section:negative}, we consider negatively correlated graphs and present achievability and converse bounds for their deanonymization.
In Section~\ref{section:conclusion}, we suggest some directions for future work.

\section{Related Work}
\label{section:related-work}

Pedarsani and Grossglauser \cite{pedarsani_privacy_2011} were the first to approach the problem of finding information-theoretic conditions for deanonymization.
Since their work, a number of authors have considered extensions and variants of the deanonymization problem.

Ji et al. \cite{ji_structural_2014} investigated the feasibility of deanonymization under the configuration model of random graphs.
The configuration model generates graphs with a specified degree sequence \cite{bollobas_probabilistic_1980}.
Real world networks differ from \ER graphs in several ways.
One of the most obvious is that ER graphs have a binomial degree distribution (which becomes approximately Poisson for sparse ER graphs), which has a rapidly decaying upper tail.
In contrast, the degree distributions of many real world networks have much heavier upper tails \cite{albert_statistical_2002}.
The configuration model allows for the replication of this feature.

Ji et al. \cite{ji_your_2015} also investigated the effect of seed information on thresholds for deanonymization.
A seed vertex pair consists of a vertex from $G_a$ and the corresponding vertex from $G_b$.
They found sufficient conditions for complete deanonymization using two information sources: first using only the edges between seed vertices and other vertices and second using all edges.
In both cases, the dependence on the number of seeds was determined.
This paper also found sufficient conditions for deanonymization of a fraction $1-\epsilon$ of the vertices.

Some practical deanonymization algorithms start by attempting to locate a few seeds.
From these seeds the and then grow the graph matching from these seeds.
Algorithms for the latter step can scale very efficiently.
Narayanan and Shmatikov were the first to apply this method \cite{narayanan_-anonymizing_2009}.
They evaluated their performance empirically on graphs derived from social networks.

More recently, there has been some work evaluating the performance of this type of algorithm on graph inputs from random models.
Yartseva and Grossglauser examined a simple percolation algorithm for growing a graph matching \cite{yartseva_performance_2013}.
They find a sharp threshold for the number of initial seeds required to ensure that final graph matching includes every vertex.
The intersection of the graphs $G_a$ and $G_b$ plays an important role in the analysis of this algorithm.
Kazemi et al. extended this work and investigated the performance of a more sophisticated percolation algorithm\cite{kazemi_growing_2015}.

If the networks being deanonymized correspond to two distinct online services, it is unlikely that the user populations of the services are identical.
Kazemi et al. investigate deanonymization of correlated graphs on overlapping but not identical vertex sets \cite{kazemi_when_2015}.
They determine that the information-theoretic penalty for imperfect overlap between the vertex sets of $G_a$ and $G_b$ is relatively mild.
This regime is an important test of the robustness of deanonymization procedures.

\section{Model}
\label{section:model}
\subsection{Notation}
For a graph $G$, let $V(G)$ and $E(G)$ be the node and edge sets respectively. 
Let $[n]$ denote the set $\{1,\cdots,n\}$.
All of the $n$-vertex graphs that we consider will have vertex set $[n]$.
This is convenient for two reasons.
First, it gives a concrete canonical way to encode the graph: take the adjacency matrix with rows and columns indexed by $[n]$.
Second, there is a clear way to define the action of a permutation on the graph.
We will always think of a permutation as a bijective function $[n] \to [n]$.
The set of permutations of $[n]$ under the binary operation of function composition form the group $S_{[n]}$.

We denote the collection of all two element subsets of $[n]$ by $\binom{[n]}{2}$.
The edge set of a graph $G$ is $E(G) \subseteq \binom{[n]}{2}$.
Let $N = \binom{n}{2} = \left|\binom{[n]}{2}\right|$.

Represent a labeled graph on the vertex set $[n]$ by its edge indicator function $g : \binom{[n]}{2} \to [2]$.
The group $S_{[n]}$ has an action on $\binom{[n]}{2}$.
This is, there is a homomorphism $l$ from $S_{[n]}$ to $S_{\binom{[n]}{2}}$.
We can write the action of the permutation $\pi$ on the graph $G$ as the composition of functions $G \circ l(\pi)$, where $l(\pi)$ is the lifted version of $\pi$:
\begin{IEEEeqnarray*}{rCl}
l(\pi) &:& \binom{[n]}{2} \to \binom{[n]}{2}\\
&& \{i,j\} \mapsto \{\pi(i), \pi(j)\}.
\end{IEEEeqnarray*}
Whenever there is only a single permutation under consideration, we will follow the convention $\sigma = l(\pi)$.

\subsection{The Deanonymization Problem}
We are considering the following problem.
There are two correlated graphs on $n$ vertices: $G_a$ and $G_b$.
By correlation we mean that for each vertex pair $e$, presence or absence of $e \in E(G_a)$, or equivalently the indicator variable $G_a(e)$, provides some information about $G_b(e)$.
The true vertex labels of $G_a$ are removed and replaced with meaningless labels.
We model this by identifying the vertices of $G_a$ with the set $[n]$ and applying a random permutation $\Pi$.
This results in the graph $G_a \circ l(\Pi)$.
The original vertex labels of $G_b$ are preserved.
We would like to know the conditions under which it is possible to discover the true correspondence between the vertices of $G_a$ and the vertices of $G_b$.
In other words, under what conditions can the random permutation $\Pi$ be recovered exactly with high probability?

In this context, an achievability result demonstrates the existence of an algorithm or estimator that exactly recovers $\Pi$ with high probability.
We will refer to such an algorithm or estimator as a deanonymizer.
A converse result is an upper bound on the probability of exact recovery that applies to any deanonymizer.

\subsection{Correlated \ER Graphs}
To fully specify this problem, we need to define a joint distribution over $G_a$ and $G_b$.
In this paper, we will focus on \ER (ER) graphs.
We have already discussed some of the advantages and drawbacks of this model in Section~\ref{section:related-work}.

We will generate correlated \ER graphs as follows.
Let $G_a$ and $G_b$ be graphs on the vertex set $[n]$.
For each $e \in \binom{[n]}{2}$, the random variables $(G_a,G_b)(e)$ are i.i.d. and
\begin{equation*}
(G_a,G_b)(e) = 
\begin{cases}
(1,1) & \text{w.p. } p_{11} \\
(1,0) & \text{w.p. } p_{10} \\
(0,1) & \text{w.p. } p_{01} \\
(0,0) & \text{w.p. } p_{00}.
\end{cases}
\end{equation*}
Call this distribution $ER(n,\bfp)$, where $\bfp = (p_{11},p_{10},p_{01},p_{00})$.
Also define the marginal probabilities for $G_a$ and $G_b$:
\begin{IEEEeqnarray*}{rCl}
p_{1*} &=& p_{11} + p_{10}\\ 
p_{0*} &=& p_{01} + p_{00}\\ 
p_{*1} &=& p_{11} + p_{01}\\ 
p_{*0} &=& p_{10} + p_{00}.
\end{IEEEeqnarray*}
Note that $G_a \sim ER(n,p_{1*})$ and $G_b \sim ER(n,p_{*1})$.

Pedarsani and Grossglauser \cite{pedarsani_privacy_2011} introduced the following generative model for correlated \ER (ER) graphs.
Essentially the same model was used in \cite{ji_structural_2014,ji_your_2015}.
Let $H$ be an ER graph on $[n]$ with edge probability $r$.
Let $G_a$ and $G_b$ be independent random subgraphs of $H$ such that each edge of $H$ appears in $G_a$ and in $G_b$ with probabilities $s_a$ and $s_b$ respectively.
We will refer to this as the \emph{subsampling model}.
The $s_a$ and $s_b$ parameters control the level of correlation between the graphs.
This is equivalent to our $ER(n,\bfp)$ model with 
\begin{IEEEeqnarray*}{rCl}
p_{11} &=& rs_as_b \\
p_{10} &=& rs_a(1-s_b) \\
p_{01} &=& r(1-s_a)s_b \\
p_{00} &=& 1-r(s_a + s_b - s_as_b).
\end{IEEEeqnarray*}
The subsampling model is capable of representing any distribution over graph pairs in which $(G_a,G_b)(e)$ are i.i.d. and the graphs have nonnegative correlation, i.e. $p_{11} \geq p_{1*}p_{*1}$.
Observe that when $G_a$ and $G_b$ are independent, $r = 1$.
The general $ER(n,\bfp)$ can represent negatively correlated graphs as well.
We will examine these in Section~\ref{section:negative}.

We will be concerned primarily with the \emph{sparse} regime, defined by $p_{1*} \to 0$ and $p_{*1} \to 0$, or equivalently $p_{00} \to 1$.
When we have the condition $\frac{p_{11}p_{00}}{p_{01}p_{10}} \rightarrow \infty$, we will say that the graphs are \emph{significantly correlated}.
Solving for $r$ from the above definitions, we obtain
\begin{equation*}
r = \frac{p_{1*}p_{*1}}{p_{11}} = p_{11} + p_{10} + p_{01} + \frac{p_{10}p_{01}}{p_{11}}.
\end{equation*}
Thus $r \rightarrow 0$ is equivalent to $p_{00} \rightarrow 1$ and $\frac{p_{11}p_{00}}{p_{01}p_{10}} \rightarrow \infty$.

In the subsampling model, it is possible to interpret $H$ as representing some ground truth and $G_a$ and $G_b$ as incomplete observations of $H$.
However, this understates the generality of this model.
Let $H$ be an $ER(n,r)$ graph and let $A$ and $B$ be the probability transition matrices of stochastic maps $[2] \to [2]$.
Let $G_a(e)$ and $G_b(e)$ be noisy observations of $H(e)$ through the channels defined by $A$ and $B$ respectively, where all channels are independent.
Then $(G_a, G_b) \sim ER(n,\bfp)$, where
\bcomment{
\begin{IEEEeqnarray*}{rCl}
p_{11} &=& rA_{1,1}B_{1,1} + (1-r)A_{0,1}B_{0,1}\\
p_{10} &=& rA_{1,1}B_{1,0} + (1-r)A_{0,1}B_{0,0}\\
p_{01} &=& rA_{1,0}B_{1,1} + (1-r)A_{0,0}B_{0,1}\\
p_{00} &=& rA_{1,0}B_{1,0} + (1-r)A_{0,0}B_{0,0}
\end{IEEEeqnarray*}}
\begin{equation*}
\begin{pmatrix} p_{0,0} & p_{0,1} \\ p_{1,0} & p_{1,1} \end{pmatrix} =
A^T \begin{pmatrix} 1-r & 0 \\ 0 & r \end{pmatrix} B.
\end{equation*}


For some $e \in \binom{[n]}{2}$, consider the correlation $\rho$ between $G_a(e)$ and $G_b(e)$.
In the sparse, significantly correlated regime we have $p_{00} \to 1$, $\frac{p_{00} p_{11}}{p_{01} p_{10}} \to \infty$, and
\begin{IEEEeqnarray*}{rCl}
\rho &=&
\frac{\E[G_a(e)G_b(e)] - \E[G_a(e)]\E[G_b(e)]}{\sqrt{(\E[G_a(e)^2] - \E[G_a(e)]^2)(\E[G_b(e)^2] - \E[G_b(e)]^2)}}\\
&=& \frac{p_{11}p_{00} - p_{01}p_{10}}{\sqrt{p_{1*}p_{0*}p_{*1}p_{*0}}}\\
&=& \frac{p_{11}}{\sqrt{p_{1*}p_{*1}}}(1+o(1))\\
&=& \sqrt{s_a s_b}(1+o(1)).
\end{IEEEeqnarray*}
Thus the intuition that in the symmetric case, $s = s_a = s_b$ measures the level of correlation between $G_a$ and $G_b$ is accurate.

Note that the condition $\frac{p_{00} p_{11}}{p_{01} p_{10}} \to \infty$ is much weaker than the condition $\rho = \Omega(1)$.
For example, consider $p_{1*} = p_{*1} = n^{-\frac{1}{2}}$ and $p_{11} = n^{-1} \log n$.
Then $\frac{p_{00} p_{11}}{p_{01} p_{10}} = \log(n)(1 + o(1))$ and $\rho = n^{-\frac{1}{2}} \log(n)(1+o(1))$.

\section{Results}
\label{section:results}
\bcomment{
\begin{figure}
\centering
\input{plot-s-r.tex}
\caption{
}
\end{figure}}

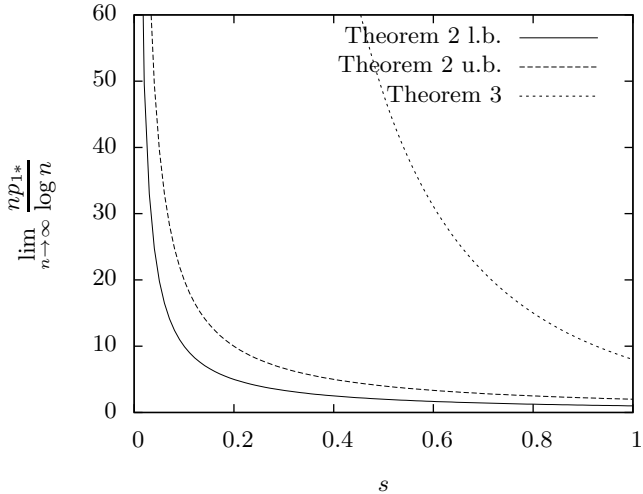
\begin{figure}
\label{fig:pos}
\centering
\input{plot-s-sr.tex}
\caption{This plot depicts the trade-off between correlation ($s = s_a = s_b$) and density ($p_{1*} = p_{*1}$) in the feasibility of exact graph matching. 
Exact recovery of the permutation is possible above the threshold given by the Theorem~\ref{thm:main} upper bound and impossible below the threshold given by the Theorem~\ref{thm:main} lower bound. 
The upper bound of Pedarsani and Grossglauser \cite{pedarsani_privacy_2011}, Theorem~\ref{thm:pg}, is also plotted for comparison.
}
\end{figure}

\bcomment{
\begin{figure}
\centering
\input{plot-s-ssr.tex}
\caption{
}
\end{figure}}

Our main achievability result applies for all regimes of $\bfp$.

\begin{theorem}
\label{thm:general-achievability}
Let $(G_a,G_b) \sim ER(n,\bfp)$, where $\bfp$ can depend on $n$.
If 
\begin{equation*}
\left(\sqrt{p_{11}p_{00}} - \sqrt{p_{01}p_{10}}\right)^2 \geq 2\frac{\log n + \omega(1)}{n},
\end{equation*}
then there is a deanonymizer that succeeds with probability $1-o(1)$.
\end{theorem}

Recall that in the sparse regime $p_{00} \to 1$ and in the significant correlation regime $\frac{p_{11}p_{00}}{p_{10}p_{01}} \to \infty$.
Here, our achievability is nearly tight.
\begin{theorem}
\label{thm:main}
Let $(G_a,G_b) \sim ER(n,\bfp)$ where $p_{00} \to 1$ and $\frac{p_{11}p_{00}}{p_{10}p_{01}} \to \infty$.
If $p_{11} \geq 2\frac{\log n + \omega(1)}{n}$, then there is a deanonymizer that succeeds with probability $1-o(1)$.
If $p_{11} \leq \frac{\log n - \omega(1)}{n}$, then any deanonymizer succeeds with probability $o(1)$.
\end{theorem}
\begin{proof}[of Theorem~\ref{thm:main}, achievability part]
This follows from Theorem~\ref{thm:general-achievability} and 
\begin{equation*}
\left(\sqrt{p_{11}p_{00}} - \sqrt{p_{01}p_{10}}\right)^2 = p_{11}p_{00}\left(1 - \sqrt{\frac{p_{10}p_{01}}{p_{11}p_{00}}}\right)^2 \to p_{11}.
\end{equation*}
\end{proof}

The achievability half of Theorem~\ref{thm:main} improves on a previous result by Pedarsani and Grossglauser \cite{pedarsani_privacy_2011}], which we restate here using our notation.
\begin{theorem}[Pedarsani and Grossglauser \cite{pedarsani_privacy_2011}]
\label{thm:pg}
Let $(G_a,G_b) \sim ER(n,\bfp)$ where $p_{10}~=~p_{01}$, $p_{00}~\to~1$, $\frac{p_{11}p_{00}}{p_{10}p_{01}}~\to~\infty$ and $\frac{p_{11}}{p_{1*}} = \omega(1/n)$.
If $\frac{p_{11}^2}{p_{10} + p_{01} + p_{11}} \geq 8\frac{\log n + \omega(1)}{n}$, then there is a deanonymizer that succeeds with probability $1-o(1)$.
\end{theorem}
The bounds in Theorems~\ref{thm:main} and \ref{thm:pg} are plotted in Figure~\ref{fig:pos}.

Theorem~\ref{thm:pg} applies to the symmetric case $p_{10}~=~p_{01}$ (and $s = s_a = s_b$).
In this case we have reduced the achievability threshold by a factor of $\frac{4(p_{10} + p_{01} + p_{11})}{p_{11}} = \frac{4(2-s)}{s}$.
This improvement becomes more significant as the graphs $G_a$ and $G_b$ become less correlated and $s$ decreases.
Additionally, the gap between the achievability and converse threshold has been reduced to a factor of 2 throughout this regime.

\subsection{Perfect Correlation Limit}
In the perfect correlation limit, i.e. $s=1$, we have $G_a = G_b$.
In this case, the size of the automorphism group of $G_a$ determines whether it is possible to recover the permutation applied to $G_a$.
This is because the composition of an automorphism with the true matching gives another matching with no errors.
Whenever the automorphism group of $G_a$ is nontrivial, it is impossible to exactly recover the permutation with high probability.
We will return to this idea in Section~\ref{section:converse} in the proof of the converse part of Theorem~\ref{thm:main}.
Wright established that for $\frac{\log n + \omega(1)}{n} \leq p \leq 1 - \frac{\log n + \omega(1)}{n}$, the automorphism group of $G \sim ER(n,p)$ is trivial with probability $1 - o(1)$ and that elsewhere, it is nontrivial with probability $1-o(1)$\cite{wright_graphs_1971}.
In fact, he proved a somewhat stronger statement about the growth rate of the number of unlabeled graphs that implies this fact about automorphism groups.

Thus for $s=1$, the converse part of Theorem~\ref{thm:main} is tight and the achievability part is off by a factor of two.
We conjecture that the converse is tight for all $s$.

Bollob\'{a}s later provided a more combinatorial proof of this automorphism group threshold function \cite{bollobas_random_1998}.
The methods we use are closer to those of Bollob\'{a}s.

\TODO{replace
\begin{itemize}
\item Explain basic strategies of both proofs
\item fix a permutation, look at random graphs
\end{itemize}}

\subsection{MAP Estimation}
The graph deanonymization problem is a statistical estimation problem.
The Maximum a Posteriori (MAP) estimator minimizes the probability of error.
The structure of the MAP estimator informs both our achievability and converse bounds.
Hence if the MAP estimator does not recover the true permutation with high probability, then no other estimator can succeed.
Note that because the permutations used to anonymize $G_a$ are equiprobable, the MAP estimator is same as the Maximum Likelihood estimator.

For two graphs on $[n]$, $G$ and $H$, let $G \cup H$ be the graph with edge set $E(G) \cup E(H)$ and let $G \cap H$ be the graph with edge set $E(G) \cap E(H)$.
Define the size of the symmetric difference of the edge sets of $G$ and $H$ as
\begin{IEEEeqnarray*}{rCl}
\Delta(G,H) &=& |E(G \cup H)| - |E(G \cap H)|\\
&=& \sum_{e \in \binom{[n]}{2}} |G(e) - H(e)|\\
\end{IEEEeqnarray*}
which is also the Hamming distance between the edge indicator vectors of $G$ and $H$. 

The MAP estimator for this problem can be derived as follows.
In the following lemma we will be careful to distinguish graph-valued random variables from fixed graphs.
Thus we name the former with upper-case letters and the latter with lower-case.
\begin{lemma}
\label{lemma:posterior}
Let $(G_a,G_b) \sim ER(n,\bfp)$, let $\Pi$ be a uniformly random permutation of $[n]$, and let $G_c = G_a \circ l(\Pi)$.
Let $k = \frac{1}{2}\Delta(g_c \circ l(\pi)^{-1}, g_b)$.
Then 
\begin{equation*}
P[\Pi = \pi|(G_c,G_b) = (g_c,g_b)] \propto \left(\frac{p_{10} p_{01}}{p_{11} p_{00}}\right)^{k}.
\end{equation*}
\end{lemma}
\begin{proof}
We compute the posterior probability as follows:
\begin{IEEEeqnarray*}{Cl}
          & P[\Pi = \pi | (G_c, G_b) = (g_c,g_b)] \\
\eqa{a}   & \frac{P[(G_c, G_b) = (g_c,g_b) | \Pi = \pi]P[\Pi = \pi]}
              {P[(G_c, G_b) = (g_c,g_b)]} \\
\propa{b} & P[(G_c, G_b) = (g_c,g_b) | \Pi = \pi] \\
\eqa{c}   & P[(G_a, G_b) = (g_c \circ l(\pi)^{-1},g_b) | \Pi = \pi] \\
\eqa{d}   & P[(G_a, G_b) = (g_c \circ l(\pi)^{-1},g_b)]
\end{IEEEeqnarray*}
where the constant of proportionality does not depend on $\pi$.
Here we have applied Bayes rule in $(a)$, the uniformity of $\Pi$ in $(b)$, the relationship between $G_a$, $G_c$ and $\Pi$ in $(c)$, and the independence of $(G_a,G_b)$ from $\Pi$ $(d)$.

Let $g_a = g_c \circ l(\pi)^{-1}$, $m_a = |E(g_a)| = |E(g_c)|$, and $m_b = |E(g_b)|$.
Then
\begin{IEEEeqnarray*}{rCl}
|\{e : (g_a,g_b)(e) = (1,1)\}| &=& \frac{m_a + m_b}{2} - k \\
|\{e : (g_a,g_b)(e) = (1,0)\}| &=& \frac{m_a - m_b}{2} + k \\
|\{e : (g_a,g_b)(e) = (0,1)\}| &=& \frac{m_b - m_a}{2} + k \\
|\{e : (g_a,g_b)(e) = (0,0)\}| &=& N - \frac{m_a + m_b}{2} - k .
\end{IEEEeqnarray*}
From the definition of the distribution of $(G_a,G_b)$, we have 
\begin{IEEEeqnarray*}{Cl}
 & P[ (G_a, G_b) = (g_a, g_b) ]\\
=& p_{11}^{\frac{m_a +m_b}{2} - k} p_{10}^{\frac{m_a - m_b}{2} + k} p_{01}^{\frac{m_b - m_a}{2} + k} p_{00}^{N - \frac{m_a + m_b}{2} - k}\\ 
\propto& \left(\frac{p_{10} p_{01}}{p_{11} p_{00}}\right)^{k}  
\end{IEEEeqnarray*}
where we have kept the factors that depends on $k$ and dropped the constant of proportionality that depends only on $m_a$ and $m_b$.
\end{proof}

Thus the entries of posterior distribution, $P[\Pi = \pi| (G_c,G_b) =(g_c,g_b)]$, depend monotonically on $\Delta(g_a,g_b)$.
If we fix any randomized estimation procedure, then the estimator $\hat{\Pi}$ is a random variable.
It will be more convenient to work with the random permutation $\hat{\Pi} \circ \Pi^{-1}$ rather than $\hat{\Pi}$ directly.
The estimator is correct when $\hat{\Pi} \circ \Pi^{-1} = I$, the identity permutation.
In fact, it is easy to see that $\hat{\Pi} \circ \Pi^{-1}$ is independent of $\Pi$.
For fixed $g_a$ and $g_b$, any change in $\Pi$ results in a corresponding change in $\hat{\Pi}$.

From here on, we do not need to consider the graph $G_a \circ l(\Pi)$.
We can work directly with $G_a$ and $G_b$ and assume that $\pi = I$ is always the correct answer to the estimation.

The following quantity is central to both our converse and our achievability arguments (as well as the achievability proof of Pedarsani and Grossglauser \cite{pedarsani_privacy_2011}).
\begin{definition}
Define $d(\pi,G_a,G_b)$ (abbreviated $d(\pi)$) to be $\Delta(G_a \circ \sigma, G_b) - \Delta(G_a, G_b)$, where $\sigma = l(\pi)$. 
\end{definition}
This is the difference in matching quality between the permutation $\pi$ and the identity permutation.

\begin{lemma}
\label{lemma:mean}
Let $\pi$ be a permutation of $[n]$, let $\sigma = l(\pi)$, and let $c_1$ be the number of fixed points of $\sigma$.
If $(G_a,G_b) \sim ER(n,\bfp)$, then $\E[d(\pi)] = 2(N-c_1)(p_{00} p_{11} - p_{01} p_{10})$.
\end{lemma}
\begin{proof}
Suppose $e \in \binom{[n]}{2}$ is not a fixed point of $\sigma$.
Then
\begin{IEEEeqnarray*}{Cl}
 & \E[|G_a(\sigma(e)) - G_b(e)| - |G_a(e) - G_b(e)|]\\
=& P[G_a(\sigma(e)) \neq G_b(e)] - P[G_a(e) \neq G_b(e)]\\
=& p_{*0} p_{1*} + p_{*1} p_{0*} - p_{01} - p_{10}\\
=& (p_{00} + p_{10}) (p_{10} + p_{11}) + (p_{01} + p_{11}) (p_{00} + p_{01})\\
 & - (p_{01} + p_{10})(p_{00} + p_{01} + p_{10} + p_{11})\\
=& 2 p_{00} p_{11} - 2 p_{01} p_{10} 
\end{IEEEeqnarray*}
and the value of $\E[d(\pi)]$ follows from linearity of expectation.
\end{proof}
Note that from Lemma~\ref{lemma:mean} the expected value of $d(\pi)$ is influenced by the number of trivial cycles (i.e. fixed points) of $\sigma$.
However, it does not depend on the distribution of lengths of the nontrivial cycles.
\begin{remark}
\label{remark}
Suppose that $\sigma = l(\pi)$ only contains cycles of length one and two.
Thus the only cycles that contribute positively to $d(\pi,G,H)$ are those containing $e,e' \in \binom{[n]}{2}$ such that $(G,H)(e) = (1,1)$ and $(G,H)(e') = (0,0)$.
The only cycles that contribute negatively to $d(\pi,G,H)$ are those containing $e,e' \in \binom{[n]}{2}$ such that $(G,H)(e) = (1,0)$ and $(G,H)(e') = (0,1)$.
\end{remark}

\section{Proof of Achievability}
\label{section:achievability}
Now we will prove the achievability half of Theorem~\ref{thm:main}.
From Lemma~\ref{lemma:posterior}, we know that the maximum a posteriori estimator is closely connected to the statistic $\Delta(G_a \circ \sigma, G_b)$.
This measures the quality of the matching produced by the permutation $\pi$.
We would like to show that with high probability, all non-identity permutations decrease the quality of the matching between $G_a$ and $G_b$.

Here is the basic strategy.
Throughout, we will analyze random graphs for some fixed permutation.
First, in Lemma~\ref{lemma:conditioning}, we will relate the distribution of $\Delta(G_a \circ \sigma, G_b)$ to $\Delta(G_a \circ \sigma, G_a)$.
In Section~\ref{subsection:cycles}, we will precisely analyze the distribution of $\Delta(G_a \circ \sigma, G_a)$.
This will allow us to obtain a tight bound on the probability that a particular permutation produces a better matching than the identity.
In Section~\ref{subsection:union}, we conclude the proof by applying the union bound over all permutations.

In will be convenient to name the following quantity from the statement of Theorem~\ref{thm:general-achievability}:
\begin{equation*}
q = \left(\sqrt{p_{11}p_{00}} - \sqrt{p_{10}p_{01}}\right)^2.
\end{equation*}

\begin{definition}
Let $R_{\pi}(z)$ be the generating function for the random variable $r = \frac{1}{2}\Delta(G_a, G_a \circ \sigma)$ where $G_a \sim ER(n,p_{1*})$:
\begin{equation*}
R_{\pi}(z) = \sum_r P[\Delta(G_a, G_a \circ \sigma) = 2r] z^r.
\end{equation*}
\end{definition}

\begin{lemma}
\label{lemma:conditioning}
Let $(G_a,G_b) \sim ER(n,\bfp)$, let $\pi$ be a permutation of $[n]$, and let $2r = \Delta(G_a, G_a \circ \sigma)$.
Conditioned on $r$, $d(\pi,G_a,G_b)$ has the generating function
\begin{equation*}
D_{r,G_b}(z) = \left(\frac{p_{00} z + p_{01} z^{-1}}{p_{0*}}\right)^r\left(\frac{p_{10} z^{-1} + p_{11} z}{p_{1*}}\right)^r.
\end{equation*}
\end{lemma}
\begin{proof}
Let $a(e) = |G_a(\sigma(e)) - G_b(e)| - |G_a(e) -  G_b(e)|$.
Then $d(\pi,G_a,G_b) = \sum_e a(e)$.
Because $a(e)$ depends on $G_a$ only at $G_a(e)$, the terms of the sum are conditionally independent.
If $G_a(\sigma(e)) = G_a(e)$, then $|G_a(\sigma(e)) - G_b(e)| = |G_a(e) -  G_b(e)|$ and the contribution of $a(e)$ to $d(\pi,G_a,G_b)$ is zero.
If $G_a(\sigma(e)) \neq G_a(e)$, then $|G_a(\sigma(e)) - G_b(e)| \neq |G_a(e) -  G_b(e)|$ and $a(e)$ is either $1$ or $-1$.

Suppose that $G_a(e) = 0$ and $G_a(\sigma(e)) = 1$.
Then 
\begin{equation*}
P[a(e) = 1 | G_a(e) = 0, G_a(\sigma(e)) = 1] = \frac{p_{00}}{p_{0*}}.
\end{equation*}
Suppose that $G_a(e) = 1$ and $G_a(\sigma(e)) = 0$.
Then 
\begin{equation*}
P[a(e) = 1 | G_a(e) = 1, G_a(\sigma(e)) = 0] = \frac{p_{11}}{p_{1*}}.
\end{equation*}
Within each cycle of $\sigma$, the number of $e$ such that $G_a(e) = 0$ and $G_a(\sigma(e)) = 1$ is equal to the number of $e$ such that $G_a(e) = 1$ and $G_a(\sigma(e)) = 0$.
Throughout all of $\sigma$, the number of $e$ such that $G_a(e) = 0$ and $G_a(\sigma(e)) = 1$ is equal to $r = \frac{1}{2}\Delta(G_a, G_a \circ \sigma)$.
Thus
\begin{equation*}
D_{r,G_b}(z) = \left(\frac{p_{00} z + p_{01} z^{-1}}{p_{0*}}\right)^r\left(\frac{p_{10} z^{-1} + p_{11} z}{p_{1*}}\right)^r.
\end{equation*}
\end{proof}

Now we will apply a standard technique to obtain tail probability bounds for large deviations from the mean.
\begin{lemma}
\label{lemma:tail-bounds}
Let $(G_a,G_b) \sim ER(n,\bfp)$, let $\pi$ be a permutation of $[n]$, and let $2r = \Delta(G_a, G_a \circ \sigma)$.

If $\E[d(\pi)] \geq 0$, then
\begin{equation*}
P[d(\pi) \leq 0] \leq R_{\pi}\left(1 - \frac{q}{p_{1*}p_{0*}}\right).
\end{equation*}
If $\E[d(\pi)] \leq 0$, then
\begin{equation*}
P[d(\pi) \geq 0] \leq R_{\pi}\left(1 - \frac{q}{p_{1*}p_{0*}}\right).
\end{equation*}
\end{lemma}
\begin{proof}
For all $0 < z \leq 1$
\begin{equation*}
P[d(\pi) \leq 0 | G_a] = \E[\one_{d(\pi) \leq 0} | G_a] \leq \E[z^{d(\pi)} | G_a] = D_{r,G_b}(z).
\end{equation*}

Starting from the expression of Lemma~\ref{lemma:conditioning}, we have
\begin{IEEEeqnarray*}{rCl}
D_{r,G_b}(z) 
&=& \left(\frac{p_{00} z + p_{01} z^{-1}}{p_{0*}}\right)^r\left(\frac{p_{10} z^{-1} + p_{11} z}{p_{1*}}\right)^r\\
&=& \left(\frac{p_{01}p_{10} z^{-2} + p_{00}p_{10} + p_{01}p_{11} + p_{11}p_{00} z^2}{p_{0*}p_{1*}}\right)^r\\
&=& \left(1 - \frac{p_{10}p_{01} + p_{11}p_{00} - p_{10}p_{01}z^{-2} - p_{11}p_{00}z^2}{p_{0*}p_{1*}}\right)^r.
\end{IEEEeqnarray*}
The value of $z$ that minimizes $D_{r,G_b}(z)$ is 
\begin{equation*}
z^* = \left(\frac{p_{01}p_{10}}{p_{00}p_{11}}\right)^{1/4}.
\end{equation*}
From Lemma~\ref{lemma:mean}, $\E[d(\pi)] = 2(N-c_1)(p_{00}p_{11} - p_{01}p_{10})$, so $z^* \leq 1$ exactly when $\E[d(\pi)] \geq 0$.
Substituting, we obtain
\begin{IEEEeqnarray*}{rCl}
  D_{r,G_b}(z^*)
&=& \left(1 - \frac{p_{10}p_{01} + p_{11}p_{00} - 2\sqrt{p_{00}p_{10}p_{01}p_{11}}}{p_{0*}p_{1*}}\right)^r\\
&=& \left(1 - \frac{\left(\sqrt{p_{11}p_{00}} - \sqrt{p_{10}p_{01}}\right)^2}{p_{0*}p_{1*}}\right)^r\\
&=& \left(1 - \frac{q}{p_{0*}p_{1*}}\right)^r\\
\end{IEEEeqnarray*}
Finally, the first claim follows from the definition of $R_{\pi}$:
\begin{IEEEeqnarray*}{rCl}
P[d(\pi) \leq 0]
&=& \sum_r P[\Delta(G_a, G_a \circ \sigma) = 2r] P[d(\pi) \leq 0|G_a]\\
&\leq& \sum_r P[\Delta(G_a, G_a \circ \sigma) = 2r] \left(1 - \frac{q}{p_{1*}p_{0*}}\right)^r\\
&=& R_{\pi}\left(1 - \frac{q}{p_{1*}p_{0*}}\right).
\end{IEEEeqnarray*}
For all $1 \leq z < \infty$, $P[d(\pi) \geq 0 | G_a] \leq D_{r,G_b}(z)$.
The proof of the second claim matches the proof of the first claim with the appropriate inequalities flipped.
\bcomment{
\begin{IEEEeqnarray*}{Cl}
 & p_{1*}p_{0*}(1-z^*)\\
=& p_{1*}p_{0*} - p_{00}p_{10} - p_{01}p_{11} - 2 \sqrt{p_{00}p_{10}p_{01}p_{11}}\\
=& p_{10}p_{00} + p_{10}p_{01} + p_{11}p_{00} + p_{11}p_{01}\\
 & -p_{00}p_{10} - p_{01}p_{11} - 2 \sqrt{p_{00}p_{10}p_{01}p_{11}}\\
=& p_{10}p_{01} + p_{11}p_{00} - 2 \sqrt{p_{00}p_{10}p_{01}p_{11}}\\
=& \left(\sqrt{p_{11}p_{00}} - \sqrt{p_{10}p_{01}}\right)^2\\
\end{IEEEeqnarray*}}
\end{proof}

\subsection{Cycle combinatorics}
\label{subsection:cycles}
Let $a_{l,k,r}$ be the number of cyclic sequences of length $l$ with $k$ ones and $r$ ones that followed by zeros.
Define the corresponding generating function 
\begin{equation*}
a_l(x,y,z) = \sum_{k,r} a_{l,k,r} x^k y^{l-k} z^r .
\end{equation*}

Let $c_l$ be the number of cycles of length $l$ in $\sigma$.
Then
\begin{equation*}
R_{\pi}(z) = \prod_{l = 1}^n a_l(p_{1*},p_{0*},z)^{c_l}
\end{equation*}
because $R_{\pi}$ is the generating function for the random variable $\Delta(G_a \circ \sigma, G_a)$ and each one followed by a zero in a cycle of $\sigma$ contributes to this quantity.

\begin{theorem}
\label{thm:Kgf}
Let $\pi$ be a permutation of $[n]$ such that $\sigma = l(\pi)$ has $c_1$ fixed points.
Then
\begin{equation*}
R_{\pi}\left(1 - \frac{q}{p_{1*}p_{0*}}\right) \leq (1 - 2q)^{\frac{N-c_1}{2}}.
\end{equation*}
\end{theorem}

The proof of Theorem~\ref{thm:Kgf} will use a few combinatorial lemmas.
Let $b_{l,s}$ be the number of cyclic sequences of length $l$ with $s$ ones, none of which are consecutive.

\begin{lemma}
\label{lemma:bijection}
For all $l,k,s \in \N$,
\begin{equation*}
\sum_r a_{l,k,r} \binom{r}{s} = b_{l,s} \binom{l-2s}{k-s}.
\end{equation*}
\end{lemma}
\begin{proof}
This identity is due the following bijection.
The left side of the equation counts cyclic sequences with $k$ ones, in which $s$ of the ones that are followed by zeros have been marked.
No two of these marked ones are consecutive.
To produce one of the objects counted on the right side, create a new cyclic sequence by placing a one each marked position and filling in the rest with zeros.
There are $b_{l,s}$ such cycles.
There are $l-2s$ remaining unspecified positions in the first cycle.
In these positions there must be $k-s$ ones and $l-k-s$ zeros.
Record the symbols at these positions in a vector.
There are $\binom{l-2s}{k-s}$ such vectors.
\end{proof}

\begin{lemma}
\label{lemma:bijection-two}
For all $l,s \in \N$.
\begin{equation*}
2^{l-2s} b_{l,s} = 2 \sum_i \binom{l}{2i} \binom{i}{s}\\
\end{equation*}
\end{lemma}
\begin{proof}
For $s \geq 1$, both sides of the equation count the set of ternary cyclic sequences of length $l$ with exactly $s$ ones, such that in each interval separating a pair of ones there are an odd number of twos (which forces the interval to be nonempty).
In such a sequence, the number of indices with either a one or a two is even.
To obtain the expression on the right side, consider the subsequence induced by these symbols and let $2i$ be its length.
In this subsequence, ones appear either only in even positions or only in odd positions, so there are $2\binom{i}{s}$ possible subsequences and $\binom{l}{2i}$ ways the subsequence can appear in the full sequences.
To obtain the expression on the left side, consider the subsequence induced by the zeros and twos.
There are $b_{l,s}$ ways this subsequence can appear in the full sequence.
Regardless of the location of the ones, there are $2^{l-2s}$ possible induced sequences of zeros and twos: there are $l-s$ total symbols broken into $s$ segments and there is a parity constraint on each segment.

For $s=0$, both sides are equal to $2^l$.
\end{proof}

\begin{lemma}
\label{lemma:cycle}
For all $l \in \N$, the formal power series $a_l(x,y,z)$ satisfies
\begin{equation*}
a_l(x,y,z) = 2^{1-l} (x+y)^l \sum_i \binom{l}{2i} \left(1 + \frac{4xy(z-1)}{(x+y)^2}\right)^i.
\end{equation*}
\end{lemma}

\begin{proof}
Applying the binomial theorem to expand $z^r$, then Lemma~\ref{lemma:bijection}, then the binomial theorem again, we obtain
\begin{IEEEeqnarray*}{rCl}
a_l(x,y,z)
&=& \sum_k \sum_r x^k y^{l-k} z^r a_{l,k,r} \\
&=& \sum_k \sum_r x^k y^{l-k} a_{l,k,r} \sum_s \binom{r}{s} (z-1)^s \\
&=& \sum_s \sum_k x^k y^{l-k} b_{l,s} \binom{l-2s}{k-s} (z-1)^s \\
&=& \sum_s b_{l,s} x^s y^s (z-1)^s \sum_k x^{k-s} y^{l-k-s}  \binom{l-2s}{k-s} \\
&=& \sum_s b_{l,s} (xy(z-1))^s (x+y)^{l-2s} \\
&=& (x+y)^l \sum_s b_{l,s} \left(\frac{xy(z-1)}{(x+y)^2}\right)^s.
\end{IEEEeqnarray*}
Applying Lemma~\ref{lemma:bijection-two} followed by the binomial theorem, we obtain
\begin{IEEEeqnarray*}{rCl}
\sum_s b_{l,s} w^s
&=& \sum_s 2^{1-l+2s} \sum_i \binom{l}{2i} \binom{i}{s} w^s\\
&=& 2^{1-l} \sum_i \binom{l}{2i} \sum_s \binom{i}{s} (4w)^s\\
&=& 2^{1-l} \sum_i \binom{l}{2i} (1+4w)^i.
\end{IEEEeqnarray*}
Combining these gives the lemma.
\end{proof}

\begin{proof}[of Theorem~\ref{thm:Kgf}]
Let $g = \sqrt{1 - 4q}$.
Substituting $p_{1*}$, $p_{0*}$, and $1 - \frac{q}{p_{1*}p_{0*}}$ into the expression from Lemma~\ref{lemma:cycle}, we obtain
\begin{IEEEeqnarray*}{rCl}
a_l\left(p_{1*},p_{0*},1 - \frac{q}{p_{1*}p_{0*}}\right)
&=& 2^{1-l} \sum_i \binom{l}{2i} g^{2i} \\
&=& 2^{-l} \sum_j \binom{l}{j} (1+(-1)^j) g^j \\
&=& \left(\frac{1+g}{2}\right)^l + \left(\frac{1-g}{2}\right)^l\\
&\leq& \left(\left(\frac{1+g}{2}\right)^2 + \left(\frac{1-g}{2}\right)^2\right)^{l/2} \\
&=& \left(\frac{1+g^2}{2}\right)^{l/2} \\
&=& (1 - 2 q)^{l/2}.
\end{IEEEeqnarray*}
Here we have used a standard p-norm inequality, which states that for a vector $\bfx$, $\|\bfx\|_l \leq \|\bfx\|_2$ when $l \geq 2$.

We have shown
\begin{equation*}
a_l\left(p_{1*}, p_{0*}, 1 - \frac{q}{p_{1*}p_{0*}}\right) \leq (1 - 2q)^{l/2}
\end{equation*}
for $l \geq 2$.
Because a cycle of length one cannot have a run boundary, $a_1(p_{1*}, p_{0*}, z) = 1$.
Combining these with 
\begin{equation*}
R_{\pi}(z) = \prod_{l = 1}^n a_l(p_{1*},p_{0*},z)^{c_l}
\end{equation*}
and $\sum_{l=1}^n l c_l = N$, we obtain the claim.
\end{proof}

\subsection{Proof of Theorem~\ref{thm:general-achievability}}
\label{subsection:union}
Now we can apply Lemma~\ref{lemma:tail-bounds} and Theorem~\ref{thm:Kgf} to prove Theorem~\ref{thm:general-achievability}.
\begin{proof}[of Theorem~\ref{thm:general-achievability}]
Let $S_{n,m}$ be the set of permutations of $[n]$ that move exactly $m$ points and fix the other $n-m$.
Then $|S_{n,m}| = \binom{n}{m} !m \leq n^m$, where $!m$ is the number of derangements of $[m]$.
If $\pi \in S_{n,m}$, then $e=\{i,j\}$ is a fixed point of $\sigma$ if either $i$ and $j$ are both fixed points of $\pi$ or $i$ and $j$ form a cycle of length 2 in $\pi$.
Thus $c_1$, the number of fixed points of $\sigma$,  satisfies $\binom{n-m}{2} \leq c_1 \leq \binom{n-m}{2} + \frac{m}{2}$.
Thus 
\begin{IEEEeqnarray}{rCl}
N - c_1 &\geq& \frac{n(n-1) - (n-m)(n-m-1)-m}{2}\nonumber\\
&=& \frac{m(2n-m-2)}{2}. \label{eq:c-bound}
\end{IEEEeqnarray}

The probability that there is some permutation that produces a better match than the identity permutation is
\begin{IEEEeqnarray*}{Cl}
& P [ \vee_{\pi \neq I} d(\pi,G_a,G_b) \leq 0 ]\\
\leq& \sum_{\pi \neq I} P [d(\pi,G_a,G_b) \leq 0]\\
=& \sum_{m=2}^n \sum_{\pi \in S_{n,m}} P [d(\pi,G_a,G_b) \leq 0]\\
\leq& \sum_{m=2}^n n^m \max_{\pi \in S_{n,m}} P [d(\pi,G_a,G_b) \leq 0].
\end{IEEEeqnarray*}
Here we applied the union bound, grouped permutations by the number of points that they move, and considered the worst case permutation in each group.


From Lemma~\ref{lemma:tail-bounds} and Theorem~\ref{thm:Kgf},
\begin{equation*}
P [d(\pi,G_a,G_b) \leq 0] \leq R_{\pi}(z^*) \leq (1 - 2q)^{\frac{N-c_1}{2}},
\end{equation*}
where $q  = \left(\sqrt{p_{11}p_{00}} - \sqrt{p_{10}p_{01}}\right)^2$.
Substituting, we obtain
\begin{IEEEeqnarray*}{Cl}
& P [ \vee_{\pi \neq I} d(\pi,G_a,G_b) \leq 0 ]\\
\leq     & \sum_{m=2}^n n^m (1 - 2q)^{\frac{N-c_1}{2}}\\
\leqa{a} & \sum_{m=2}^n n^m (1 - 2q)^{\frac{m(2n-m-2)}{4}}\\
=        & \sum_{m=2}^n \left(n \exp \left( \frac{2n-m-2}{4} \log (1 - 2q) \right)\right)^m \\
\leqa{b} & \sum_{m=2}^n \left(n \exp \left( \frac{-q(2n-m-2)}{2} \right)\right)^m \\
\leqa{c} & \sum_{m=2}^n \left(n \exp \left( -\frac{q(n-2)}{2} \right)\right)^m.
\end{IEEEeqnarray*}
Inequality $(a)$ follows from \eqref{eq:c-bound}, inequality $(b)$ follows from $\log(1+x) \leq x$, and inequality $(c)$ follows from $m \leq n$.

Let $x = n \exp(-q(n-2)/2)$.
The condition $x = o(1)$ is equivalent to
\begin{equation*}
q \geq 2\frac{\log (n \omega(1))}{n-2} = 2\frac{\log n + \omega(1)}{n},
\end{equation*}
which is exactly a hypothesis of the theorem.
Thus for sufficiently large $n$, we have $x < 1$ and 
\begin{equation*}
P [ \vee_{\pi \neq I} d(\pi,G_a,G_b) \leq 0 ] \leq \sum_{m=2}^n x^m < \frac{x^2}{1-x} = o(1).
\end{equation*}
\end{proof}

\section{Proof of Converse}
\label{section:converse}
The converse statement depends on the following lemma.
\begin{lemma}
\label{lemma:intersection}
Let $G_a$ and $G_b$ be graphs on the vertex set $[n]$.
For all $\pi \in Aut(G_a \cap G_b)$, $d(\pi,G_a,G_b) \leq 0$.
\end{lemma}
\begin{proof}
Let $\sigma = l(\pi)$ and recall that
\begin{IEEEeqnarray*}{rCl}
\Delta(G,H) &=& \sum_{e \in \binom{[n]}{2}} |G(e) - H(e)|\\
d(\pi,G_a,G_b) &=& \Delta(G_a \circ \sigma, G_b) - \Delta(G_a, G_b).
\end{IEEEeqnarray*}

Let $e \in \binom{[n]}{2}$.
Suppose that $(G_a,G_b)(e) = (1,1)$, so $(G_a \cap G_b)(e) = 1$.
Because $\pi \in Aut(G_a \cap G_b)$, $(G_a \cap G_b)(\sigma(e)) = 1$.
Then the contribution of $e$ 
to both $\Delta(G_a,G_b)$ and $\Delta(G_a \circ \sigma, G_b)$ is zero.

Suppose $(G_a \cap G_b)(e) = 0$.
The cycle of $\sigma$ containing $e$ is $C = \{\sigma^i(e) : i \in \N\}$.
For all $e' \in C$, $(G_a \cap G_b)(e') = 0$ and $(G_a, G_b)(e')$ is $(0,0)$, $(0,1)$, or $(1,0)$.
Thus the contribution of $C$ to $\Delta(G_a,G_b)$ is equal to total number of edges in $G_a$ and $G_b$ in $C$.
The contribution of $C$ to $\Delta(G_a \circ \sigma, G_b)$ cannot be larger.
\end{proof}

It is well-known that \ER graphs with average degree less than $\log n$ have many automorphisms \cite{bollobas_random_1998}.
The following lemma is precise version of this fact that is suitable for our purposes.
\begin{lemma}
\label{lemma:isolated}
Let $G \sim ER(n,p)$.
If $p \leq \frac{\log n - c_n}{n}$ and $c_n \to \infty$, then there is some sequence $\epsilon_n \to 0$ such that $P[|Aut(G)| \leq \epsilon_n^{-1}] \leq \epsilon_n$.

\end{lemma}
\begin{proof}
Let $X$ be the number of isolated vertices in $G$.
A permutation that moves only isolated vertices is an automorphism of $G$, so $|Aut(G)| \geq X!$.
We will use Chebyshev's inequality to bound the probability that there are few isolated vertices in $G$: 
\begin{equation*}
P[X \leq \frac{1}{2} \E[X]] \leq 4\frac{\E[X^2] - \E[X]^2}{\E[X]^2}.
\end{equation*}

The probability that a particular vertex is isolated is $(1-p)^{n-1}$.
Thus $\E[X] = n (1-p)^{n-1}$.
The probability that a particular pair of vertices are both isolated is $(1-p)^{2n-3}$.
Thus $E\left[\binom{X}{2}\right] = \binom{n}{2} (1-p)^{2n-3}$.
Then
\begin{IEEEeqnarray*}{Cl}
 & \E[X^2]\E[X]^{-2} - 1\\
=& \left(2 E\left[\binom{X}{2}\right] + \E[X]\right)\E[X]^{-2} - 1\\
=& \frac{(n^2-n) (1-p)^{2n-3}}{n^2 (1-p)^{2n-2}} + E[X]^{-1} - 1\\
=& (1-p)^{-1} - n^{-1}(1-p)^{-1} + E[X]^{-1} - 1\\
\leq& p + E[X]^{-1}.
\end{IEEEeqnarray*}
Recall that $p \leq \frac{\log n}{n}$, so $p \to 0$.
Finally we compute the limiting behavior of the expected value of $X$:
\begin{IEEEeqnarray*}{rCl}
\E[X] 
&=& n(1-p)^{n-1}\\
&=& n\left(1 + \frac{p}{1-p}\right)^{-(n-1)}\\
&\geq& n\left(\exp\left(\frac{p}{1-p}\right)\right)^{-n}\\
&=& \exp\left(\log n - \frac{np}{1-p}\right)\\
&=& \exp\left(\frac{c_n - p \log n}{1-p}\right).
\end{IEEEeqnarray*}
Note that $p \log n \to 0$ and $c_n \to \infty$, so $\E[X] \to \infty$.
Thus $P[X \leq \frac{1}{2} \E[X]] \to 0$. 
\end{proof}

\begin{proof}[of Theorem~\ref{thm:main}, converse part]
For all sufficiently large $n$, we have $\frac{p_{11}p_{00}}{p_{10}p_{01}} > 1$, so from Lemma~\ref{lemma:posterior}, if $\Delta(G_a,G_b) \geq \Delta(G_a, G_b \circ \sigma)$, then the posterior probability of $\pi$ is at least as large as the true permutation. 
From Lemma~\ref{lemma:intersection}, there are at least $|Aut(G_a \cap G_b)|$ such permutations.
Thus any estimator for $\Pi$ succeeds with probability at most $1/|Aut(G_a \cap G_b)|$.
The graph $G_a \cap G_b$ is distributed as $ER(n,p_{11})$.
With high probability, the size of the automorphism group of an $ER(n,p_{11})$ graph goes to infinity with $n$.
More precisely, if $p_{11} \leq \frac{\log n - \omega(1)}{n}$, then from Lemma~\ref{lemma:isolated} there is some sequence $\epsilon_n \to 0$ such that 
\begin{equation*}
P\left[\frac{1}{|Aut(G_a \cap G_b)|} \geq \epsilon_n\right] \leq \epsilon_n.
\end{equation*}
Any estimator succeeds with probability at most $2\epsilon_n$.
\end{proof}

\section{Negative Correlation}
\label{section:negative}
In this section, we consider the problem of deanonymizing negatively correlated \ER graph.
Such a pair of graphs have fewer edges in common that an independently generated pair.
In the most extreme case the graphs avoid each other completely.
Thus it is somewhat surprising that almost the same analysis yields an achievability result for both the positively and negatively correlated regimes.

In the sparse regime with significant negative correlation, we have $p_{00} \to 1$, $\frac{p_{00} p_{11}}{p_{01} p_{10}} \to 0$, and
\begin{IEEEeqnarray*}{rCl}
-\rho
&=& \frac{-p_{11}p_{00} + p_{01}p_{10}}{\sqrt{p_{1*}p_{0*}p_{*1}p_{*0}}}\\
&=& \frac{p_{01}p_{10}}{\sqrt{p_{1*}p_{*1}}}(1+o(1))\\
&\leq& \sqrt{p_{1*}p_{*1}}(1+o(1))\\
&=& o(1).
\end{IEEEeqnarray*}
Thus it is impossible to achieve the same level of correlation as in the positive case.
To counteract this, much higher edge densities are required to make deanonymization feasible.

The diameter of a graph is the maximum distance between a pair of vertices.
As we have shown in Section~\ref{section:converse}, for positively correlated graph pairs, the threshold for perfect recovery of the permutation is closely related to the threshold for a single ER graph being connected, i.e. having a finite diameter.
For negatively correlated pairs, the recovery threshold is related to the threshold for a single ER graph having a diameter of two.


The following lemma is analogous to Lemma~\ref{lemma:intersection}.
Let $G_a \setminus G_b$ be the graph with edge set $E(G_a) \setminus E(G_b)$.
\begin{lemma}
\label{lemma:difference}
Let $\pi$ be a permutation of $[n]$ that contains only cycles of length one and two.
If $d(\pi, G_a \setminus G_b, G_b \setminus G_a) = 0$, then $d(\pi,G_a,G_b) \geq 0$.
\end{lemma}
\begin{proof}
The permutation $\sigma = l(\pi)$ also contains only cycles of length one and two.
Recall Remark~\ref{remark}.
Thus the only cycles that contribute positively to $d(\pi,G,H)$ are those containing $e,e' \in \binom{[n]}{2}$ such that $(G,H)(e) = (1,1)$ and $(G,H)(e') = (0,0)$.
The only cycles that contribute negatively to $d(\pi,G,H)$ are those containing $e,e' \in \binom{[n]}{2}$ such that $(G,H)(e) = (1,0)$ and $(G,H)(e') = (0,1)$.

$(G_a \setminus G_b, G_b \setminus G_a)(e) = (0,0)$ for all $e$ such that $(G_a,G_b)(e) = (1,1)$.
Elsewhere $(G_a \setminus G_b, G_b \setminus G_a)(e) = (G_a,G_b)(e)$.
In $\sigma$, there are no cycles that contribute positively to $d(\pi,G_a \setminus G_b, G_b \setminus G_a)$.
Because $d(\pi,G_a \setminus G_b, G_b \setminus G_a) = 0$, there are also no cycles that contribute negatively.
Thus there are also no cycles that contribute negatively to $d(\pi,G_a,G_b) \geq 0$.
\end{proof}

Let $N_a(u)$ be the neighborhood of vertex $u$ in $G_a$.
\begin{lemma}
Let $G_a$ and $G_b$ be graphs on $[n]$ such that $E(G_a \cap G_b) = \varnothing$.
Let $u,v \in [n]$ be distinct vertices.
Let $\pi$ be the permutation that exchanges $u$ and $v$ and fixes the rest of $[n]$.
If $N_a(u) \cap N_b(v) = \varnothing$ and $N_a(v) \cap N_b(u) = \varnothing$, then $d(\pi, G_a, G_b) = 0$.
\end{lemma}
\begin{proof}
For all $i \in [n] \setminus \{u,v\}$, $\sigma$ exchanges $\{u,i\}$ with $\{v,i\}$ and fixes all other vertex pairs.
Then $\sigma$ aligns two edges if there is some $i$ such that either $G_a(\{u,i\}) = 1$ and $G_b(\{v,i\}) = 1$ or $G_a(\{v,i\}) = 1$ and $G_b(\{u,i\}) = 1$.
Thus $d(\pi,G_a, G_b) =  - |N_a(u) \cap N_b(v)|  - |N_a(v) \cap N_b(u)| = 0$.
\end{proof}

\begin{lemma}
\label{lemma:pair-lb}
Let $(G_a,G_b) \sim ER(n,\bfp)$ with $p_{11} = 0$.
If $p_{01}p_{10} \leq \frac{\log n - \omega(1)}{n}$, then with probability $1-o(1)$ there are $\omega(1)$ vertex pairs $u,v \in [n]$ such that $|N_a(u) \cap N_b(v)| = 0$ and $|N_a(v) \cap N_b(u)| = 0$.
\end{lemma}
\begin{proof}
Let $u,v \in [n]$ and let $\pi$ be the permutation that exchanges $u$ and $v$ and fixes the rest of $[n]$.
If $N_a(u) \cap N_b(v) = \varnothing$ and $N_a(v) \cap N_b(u) = \varnothing$, then $d(\pi, G_a, G_b) = 0$.

For each $i \in [n] \setminus \{u,v\}$, 
\begin{equation*}
P[G_a(\{u,i\}) = 1 \wedge G_b(\{v,i\}) = 1] = p_{1*} p_{*1} = p_{10} p_{01}.
\end{equation*}

Thus the probability that $|N_a(u) \cap N_b(v)| = 0$ and $|N_a(v) \cap N_b(u)| = 0$ is $(1 - 2 p_{10} p_{01})^{n-2}$.
Let $X$ be the number of vertex pairs satisfying this property.
Then $\E[X] = \binom{n}{2}(1 - 2p_{10} p_{01})^{n-2}$.
Suppose that $p_{10} p_{01} \leq \frac{\log n - c_n}{n}$ where $c_n \to \infty$.
A computation very similar to the one in Lemma~\ref{lemma:isolated} shows that
\begin{equation*}
\E[X] \geq \frac{1-o(1)}{2} \exp\left(2 \frac{c_n - p_{10} p_{01} \log n}{1- p_{10} p_{01}}\right)
\end{equation*}
Thus $\E[X] \to \infty$.
Another computation similar to that of Lemma~\ref{lemma:isolated} shows that 
$\frac{\E[X^2] - \E[X]^2}{\E[X]^2} = o(1)$.
Then Chebyshev's inequality implies that $P[X \leq \frac{1}{2} \E[X]] = o(1)$.
\end{proof}

\begin{theorem}
\label{thm:negative}
Let $(G_a,G_b) \sim ER(n,\bfp)$ where $\frac{p_{11}p_{00}}{p_{10}p_{01}} \to 0$.
If $p_{01}p_{10} \geq 2\frac{\log n + \omega(1)}{n}$, then there is a deanonymizer that succeeds with probability $1-o(1)$.
If $p_{01}p_{10} \leq \frac{\log n - \omega(1)}{n}$, then any deanonymizer succeeds with probability $o(1)$.
\end{theorem}
\begin{proof}
The conditions $p_{00} \to 1$ and $\frac{p_{11}p_{00}}{p_{10}p_{01}} \to 0$ imply $q \to p_{01}p_{10}$.
Thus the first statement follows from Theorem~\ref{thm:general-achievability}.

For sufficiently large $n$, we have $\frac{p_{11}p_{00}}{p_{10}p_{01}} < 1$.
From Lemma~\ref{lemma:posterior}, if $\Delta(G_a,G_b) \leq \Delta(G_a, G_b \circ \sigma)$, then the posterior probability of $\pi$ is at least as large as the true permutation. 
The statement of Lemma~\ref{lemma:difference} identifies some permutation satisfying this condition.
From Lemma~\ref{lemma:pair-lb}, with probability $1-o(1)$ there are $\omega(1)$ such permutations.
Thus any deanonymizer succeeds with probability at most $o(1)$.
\end{proof}

\section{Conclusion}
\label{section:conclusion}
We obtained new achievability and converse bounds for the problem of exact deanonymization of \ER graphs.
For sparse, significantly correlated graphs, these bounds differ by a constant factor of two.
Thus they establish the dependence of the threshold for deanonymization on the level of correlation between the graphs.
We also investigated deanonymization of negatively correlated graphs.
In sparse graphs, because it is impossible to achieve levels of negative correlation that match the possible levels of positive correlation, much higher edge density is required to facilitate deanonymization.

There are several directions that this work could be extended.
For positive correlation, the perfect correlation limit suggests that the converse bound is tight and the achievability can be further improved.
Thus the most obvious next step is to attempt to improve the achievability result to match the converse.
In the analysis of the automorphism group of an \ER graph, different methods were required to handle the region just above the threshold.
It may be possible to adapt these to the more general deanonymization problem.

Several follow-up works have applied the methods from the original paper of Pedarsani and Grossglauser \cite{pedarsani_privacy_2011}) to variations on the deanonymization problem.
We discuss some of these in Section~\ref{section:related-work}.
It is likely that our results and methods can lead to improvements in the results obtained for these extensions.

In this paper, we have focused on exact recovery of the true graph matching.
There are a number of natural relaxations of this condition.
In many of the motivating problems related to data privacy, partial information leakage is still a serious issue.
It remains to be seen which metrics for partial deanonymization are both relevant to practical privacy problems and tractable to analyze.

\bibliographystyle{abbrv}
\bibliography{sigproc,deanon}

\end{document}

%% file: plot-s-sr.tex
\begingroup
  \makeatletter
  \providecommand\color[2][]{%
    \GenericError{(gnuplot) \space\space\space\@spaces}{%
      Package color not loaded in conjunction with
      terminal option `colourtext'%
    }{See the gnuplot documentation for explanation.%
    }{Either use 'blacktext' in gnuplot or load the package
      color.sty in LaTeX.}%
    \renewcommand\color[2][]{}%
  }%
  \providecommand\includegraphics[2][]{%
    \GenericError{(gnuplot) \space\space\space\@spaces}{%
      Package graphicx or graphics not loaded%
    }{See the gnuplot documentation for explanation.%
    }{The gnuplot epslatex terminal needs graphicx.sty or graphics.sty.}%
    \renewcommand\includegraphics[2][]{}%
  }%
  \providecommand\rotatebox[2]{#2}%
  \@ifundefined{ifGPcolor}{%
    \newif\ifGPcolor
    \GPcolorfalse
  }{}%
  \@ifundefined{ifGPblacktext}{%
    \newif\ifGPblacktext
    \GPblacktexttrue
  }{}%
  \let\gplgaddtomacro\g@addto@macro
  \gdef\gplbacktext{}%
  \gdef\gplfronttext{}%
  \makeatother
  \ifGPblacktext
    \def\colorrgb#1{}%
    \def\colorgray#1{}%
  \else
    \ifGPcolor
      \def\colorrgb#1{\color[rgb]{#1}}%
      \def\colorgray#1{\color[gray]{#1}}%
      \expandafter\def\csname LTw\endcsname{\color{white}}%
      \expandafter\def\csname LTb\endcsname{\color{black}}%
      \expandafter\def\csname LTa\endcsname{\color{black}}%
      \expandafter\def\csname LT0\endcsname{\color[rgb]{1,0,0}}%
      \expandafter\def\csname LT1\endcsname{\color[rgb]{0,1,0}}%
      \expandafter\def\csname LT2\endcsname{\color[rgb]{0,0,1}}%
      \expandafter\def\csname LT3\endcsname{\color[rgb]{1,0,1}}%
      \expandafter\def\csname LT4\endcsname{\color[rgb]{0,1,1}}%
      \expandafter\def\csname LT5\endcsname{\color[rgb]{1,1,0}}%
      \expandafter\def\csname LT6\endcsname{\color[rgb]{0,0,0}}%
      \expandafter\def\csname LT7\endcsname{\color[rgb]{1,0.3,0}}%
      \expandafter\def\csname LT8\endcsname{\color[rgb]{0.5,0.5,0.5}}%
    \else
      \def\colorrgb#1{\color{black}}%
      \def\colorgray#1{\color[gray]{#1}}%
      \expandafter\def\csname LTw\endcsname{\color{white}}%
      \expandafter\def\csname LTb\endcsname{\color{black}}%
      \expandafter\def\csname LTa\endcsname{\color{black}}%
      \expandafter\def\csname LT0\endcsname{\color{black}}%
      \expandafter\def\csname LT1\endcsname{\color{black}}%
      \expandafter\def\csname LT2\endcsname{\color{black}}%
      \expandafter\def\csname LT3\endcsname{\color{black}}%
      \expandafter\def\csname LT4\endcsname{\color{black}}%
      \expandafter\def\csname LT5\endcsname{\color{black}}%
      \expandafter\def\csname LT6\endcsname{\color{black}}%
      \expandafter\def\csname LT7\endcsname{\color{black}}%
      \expandafter\def\csname LT8\endcsname{\color{black}}%
    \fi
  \fi
  \setlength{\unitlength}{0.0500bp}%
  \begin{picture}(5102.00,3968.00)%
    \gplgaddtomacro\gplbacktext{%
      \csname LTb\endcsname%
      \put(814,704){\makebox(0,0)[r]{\strut{} 0}}%
      \put(814,1204){\makebox(0,0)[r]{\strut{} 10}}%
      \put(814,1704){\makebox(0,0)[r]{\strut{} 20}}%
      \put(814,2204){\makebox(0,0)[r]{\strut{} 30}}%
      \put(814,2703){\makebox(0,0)[r]{\strut{} 40}}%
      \put(814,3203){\makebox(0,0)[r]{\strut{} 50}}%
      \put(814,3703){\makebox(0,0)[r]{\strut{} 60}}%
      \put(946,484){\makebox(0,0){\strut{} 0}}%
      \put(1698,484){\makebox(0,0){\strut{} 0.2}}%
      \put(2450,484){\makebox(0,0){\strut{} 0.4}}%
      \put(3201,484){\makebox(0,0){\strut{} 0.6}}%
      \put(3953,484){\makebox(0,0){\strut{} 0.8}}%
      \put(4705,484){\makebox(0,0){\strut{} 1}}%
      \put(176,2203){\rotatebox{-270}{\makebox(0,0){\strut{}$\displaystyle \lim_{n \to \infty} \frac{n p_{1*}}{\log n}$}}}%
      \put(2825,154){\makebox(0,0){\strut{}$s$}}%
    }%
    \gplgaddtomacro\gplfronttext{%
      \csname LTb\endcsname%
      \put(3718,3530){\makebox(0,0)[r]{\strut{}Theorem~\ref{thm:main} l.b.}}%
      \csname LTb\endcsname%
      \put(3718,3310){\makebox(0,0)[r]{\strut{}Theorem~\ref{thm:main} u.b.}}%
      \csname LTb\endcsname%
      \put(3718,3090){\makebox(0,0)[r]{\strut{}Theorem~\ref{thm:pg}}}%
    }%
    \gplbacktext
    \put(0,0){\includegraphics{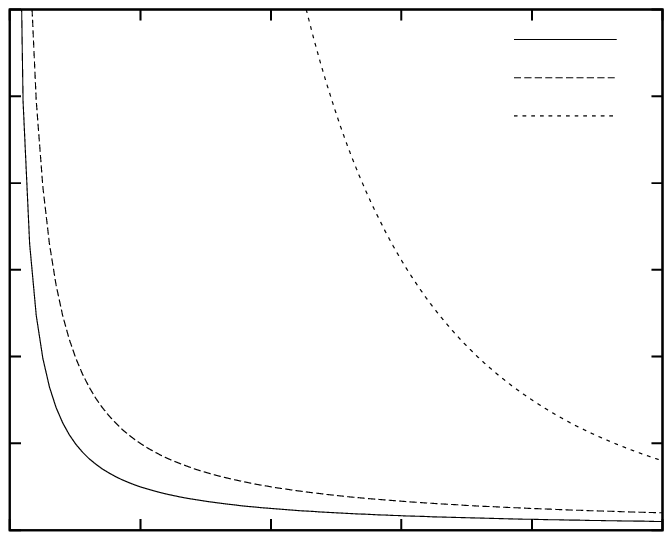}}%
    \gplfronttext
  \end{picture}%
\endgroup